\documentclass[a4paper]{article}
\usepackage{amsthm,amsmath,amssymb,fullpage}

\newtheorem{thm}{Theorem}[section]
\newtheorem{prop}[thm]{Proposition}
\newtheorem{lem}[thm]{Lemma}
\newtheorem{cor}[thm]{Corollary}

\newtheorem{rem}[thm]{Remark}
\newtheorem{conj}[thm]{Conjecture}
\numberwithin{equation}{section}

\allowdisplaybreaks

\title{Poisson structures for a similarity reduction of the Drinfeld-Sokolov hierarchy of type $A$}
\author{Takao Suzuki \thanks{Department of Mathematics, Kindai University, 3-4-1 Kowakae, Higashi-Osaka, Osaka 577-8502, Japan. E-mail: suzuki@math.kindai.ac.jp}}
\date{}

\begin{document}

\maketitle

\begin{abstract}
In this article, we consider the Drinfeld-Sokolov hierarchy of type $A$ and its similarity reduction.
The similarity reduction is expressed as a Hamiltonian system in terms of canonical coordinates.
We also investigate the connection between the similarity reduction and isomonodromy systems.


\end{abstract}

\section{Introduction}

The Drinfeld-Sokolov (DS) hierarchies were first proposed in \cite{DS} as extensions of the KdV hierarchy for affine Lie algebras.
At that time, the principal Heisenberg subalgebras of affine Lie algebras were used to formulate the hierarchies.
Afterward, the DS hierarchies were generalized to arbitrary Heisenberg subalgebras in \cite{GHM}.
On the other hand, it was shown in \cite{KP} that the isomorphism classes of the Heisenberg subalgebras are in one-to-one correspondence with the conjugacy classes of the finite Weyl groups.
Moreover, the conjugacy classes of the finite Weyl groups were classified in \cite{C}.
Hence, we can classify the DS hierarchies.
However, it is another problem to give explicit formulas for the Heisenberg subalgebras.
For classical types, a method of formulation using the conjugacy classes of the finite Weyl groups was established in \cite{DF}.
In particular, for type $A$, a formulation using block partitions of matrices was given in \cite{KL}.
Both formulations naturally induce $\mathbb{Z}$-gradations of the affine Lie algebras.

\[\begin{array}{|c|c|c|c|}\hline
	\text{Lie algebra} & \text{Conjugacy Class} & \text{Painlev\'{e}-type system} & \text{Reference} \\\hline\hline
	A_1 & (2) & P_{\rm II} & \cite{AS} \\[2pt]
	& (1,1) & P_{\rm IV} & \cite{KK1} \\[2pt]\hline
	A_2 & (3) & P_{\rm IV} & \cite{NouY2} \\[2pt]
	& (2,1) & P_{\rm V} & \cite{KIK} \\[2pt]
	& (1,1,1) & P_{\rm VI} & \cite{KK2} \\[2pt]\hline
	A_{2n-1} & (2n) & P(A_{2n-1}) & \cite{NouY2} \\[2pt]
	(n\geq2) & (2n-1,1) & P(A_{2n}) & \cite{FS4,Suz1} \\[2pt]
	& (n,n) & P_{(n,n)} & \cite{FS4,Suz1} \\[2pt]
	& (1^{2n}) & \text{Garnier} & \cite{KK2} \\[2pt]\hline
	A_{2n} & (2n+1) & P(A_{2n}) & \cite{NouY2} \\[2pt]
	(n\geq2)& (2n,1) & P(A_{2n+1}) & \cite{FS4,Suz1} \\[2pt]
	& (n,n,1) & P_{(n+1,n+1)} & \cite{FS4,Suz1} \\[2pt]
	& (1^{2n+1}) & \text{Garnier} & \cite{KK2} \\[2pt]\hline
	D^{(1)}_4 & (\bar{2},\bar{2}) & P_{\rm{VI}} & \cite{FS1} \\[2pt]\hline
	D^{(1)}_{2n+2} & (\overline{n+1},\overline{n+1}) & \text{Sasano} & \cite{FS2} \\[2pt]\hline
	E^{(1)}_6 & E_6(a_1) &  & \cite{FS3} \\[2pt]\hline
\end{array}\]

Some of the DS hierarchies reduce to Painlev\'e equations $P_{\rm J}\ ({\rm J}={\rm II},\ldots,{\rm VI})$ and their generalizations via an operation called a similarity reduction.
We list the known connections between the DS hierarchies and the Painlev\'e-type systems in the above table.
We give some remarks about this table.
The systems named $P(A_{2n})$ and $P(A_{2n+1})$ are generalizations of $P_{\rm IV}$ and $P_{\rm V}$, respectively.
They were independently proposed by several researchers in \cite{A,NouY1,Sch,VS}.
The system named $P_{(n,n)}$ is a generalization of $P_{\rm VI}$, which corresponds to the case $n=2$.
The system called the Sasano system was first proposed in \cite{Sas} as a generalization of $P_{\rm VI}$ from the viewpoint of the symmetry and the initial value space.

In this article, we consider the DS hierarchy of type $A_{(m+1)(n+1)-1}$ corresponding to the conjugacy class $(n+1,\ldots,n+1)$.
Formulations of the DS hierarchies and their similarity reductions have been largely established by previous works.
However, a universal method to express the similarity reductions as Hamiltonian systems has not yet been established.
Moreover, the resulting systems generally contain some invariants, which must be identified case by case.
We address these problems for the conjugacy class $(n+1,\ldots,n+1)$.
As a result, the similarity reduction is expressed as a Hamiltonian system in terms of $m(m+1)(n+1)$ canonical coordinates.

\[\begin{array}{|c|c|c|}\hline
	\text{Spectral Type} & \text{Painlev\'{e}-type system} & \text{Reference} \\\hline\hline
	11,11,11,11 & P_{\rm VI} & \cite{F} \\\hline
	11,11,11,11,11 & \text{Garnier} & \cite{G} \\
	21,21,111,111 & \mathcal{H}_{3,1} & \cite{T2} \\
	22,22,22,211 & \text{Matrix} & \cite{Sak} \\
	31,22,22,1111 & \text{Sasano} & \cite{Sak} \\\hline
	11,11,11,11,11,11 & \text{Garnier} & \cite{G} \\
	21,21,21,21,111 &  & \cite{Suz2} \\
	31,31,22,22,22 &  & \cite{Suz2} \\
	21,111,111,111 &  & \cite{Suz2} \\
	22,22,211,211 &  & \cite{Suz2} \\
	22,22,22,1111 &  & \cite{Suz2} \\
	31,31,1111,1111 & \mathcal{H}_{4,1} & \cite{T2} \\
	31,22,211,1111 &  & \cite{Suz2} \\
	33,33,33,321 & \text{Matrix} & \cite{K} \\
	42,33,33,222 &  & \cite{Suz2} \\
	51,33,222,222 &  & \cite{Suz2} \\
	51,33,33,111111 & \text{Sasano} & \cite{FISS} \\\hline
\end{array}\]

Another aim of this article is to clarify the connection between the DS hierarchies and the isomonodromy systems.
Recently, a classification theory of isomonodromy systems has been established.
It was shown in \cite{Kos,O} that irreducible Fuchsian systems with a fixed number of accessory parameters can be reduced to finitely many types via addition and middle convolution.
It was also shown in \cite{HF} that the isomonodromy system of the Fuchsian system is invariant under the addition and the middle convolution.
Hence, we can classify the isomonodromy systems in terms of the spectral type, which is the multiplicative data of the eigenvalues of residue matrices.
We list the isomonodromy systems with six or fewer accessory parameters in the above table.
We give some remarks about this table.
The system named $\mathcal{H}_{n,1}$ coincides with the system $P_{(n,n)}$, and both were discovered independently of each other.
The system $\mathcal{H}_{3,1}$ was also obtained in \cite{Sak}.
The system called the matrix Painlev\'e system is another generalization of $P_{\rm VI}$.
It was also independently proposed in \cite{B} and is now regarded as an example of a non-commutative Painlev\'e equation.
The system with spectral type $51,33,222,222$ coincides with that derived from the DS hierarchy of type $E^{(1)}_6$ corresponding to the conjugacy class $E_6(a_1)$.
In this way, the connection between the DS hierarchies and the isomonodromy systems has been partially clarified.
Our ultimate goal is to establish the complete connection.
In this article, we focus on the similarity reduction for the case $(m,n)=(2,1)$ and reduce it to the isomonodromy system with spectral type $31,22,211,1111$ by using a Laplace transformation.

\begin{rem}
In \cite{T1,T2}, the system $\mathcal{H}_{n,1}$ was derived from the UC hierarchy, which is an extension of the KP hierarchy for the irreducible characters of the general linear group.
Furthermore, in \cite{T1}, the systems $P(A_{2n})$ and $P(A_{2n+1})$ were also derived from the UC hierarchies.
\end{rem}

\begin{rem}
In \cite{LWZ}, several Painlev\'e-type systems were derived from the DS hierarchies via constraints known as the Virasoro constraints.
This is an important result that contributes to the connection between the Painlev\'e-type systems and conformal field theory.
Therefore, it is important to clarify the connection between this previous work and the present article, but this remains a subject for future research.
\end{rem}

This article is organized as follows.
In Section \ref{Sec:AffLie}, we recall the definition of the affine Lie algebra of type $A$.
In Section \ref{Sec:DS}, the DS hierarchy of type $A_{(m+1)(n+1)-1}$ corresponding to the conjugacy class $(n+1,\ldots,n+1)$ and its similarity reduction are formulated.
We also transform the similarity reduction to that on a Borel subalgebra by a gauge transformation.
In Section \ref{Sec:Ham}, the similarity reductions are expressed as Hamiltonian systems.
Ultimately, we arrive at a Hamiltonian system expressed in terms of canonical coordinates.
In Section \ref{Sec:IMDS}, we derive the isomonodromy system with spectral type $31,22,211,1111$.

\section{Affine Lie algebra of type $A$}\label{Sec:AffLie}

Let $m,n\in\mathbb{Z}_{>0}$ and $N=(m+1)(n+1)-1$.
We consider the affine Lie algebra $\mathfrak{g}=\mathfrak{g}(A^{(1)}_N)$.
The generalized Cartan matrix $A=\left[a_{i,j}\right]_{i,j=0}^{N}$ for $\mathfrak{g}$ is given by
\[\begin{array}{ll}
	a_{i,i}=2 & (i=0,\ldots,N), \\[4pt]
	a_{i,i+1}=a_{N,0}=a_{i+1,i}=a_{0,N}=-1 & (i=0,\ldots,N-1), \\[4pt]
	a_{i,j}=0 & (\text{otherwise}).
\end{array}\]
The Chevalley generators $\alpha^{\vee}_i,e_i,f_i\ (i=0,\ldots,N)$ and the scaling element $d$ of $\mathfrak{g}$ satisfy the fundamental relations
\begin{align*}
	&(\mathrm{ad}\,e_i)^{1-a_{i,j}}(e_j) = 0,\quad
	(\mathrm{ad}\,f_i)^{1-a_{i,j}}(f_j) = 0\quad (i\neq j), \\
	&[\alpha^{\vee}_i,\alpha^{\vee}_j] = 0,\quad
	[\alpha^{\vee}_i,e_j] = a_{i,j}e_j,\quad
	[\alpha^{\vee}_i,f_j] = -a_{i,j}f_j,\quad
	[e_i,f_j] = \delta_{i,j}\alpha^{\vee}_i, \\
	&[d,\alpha^{\vee}_i] = 0,\quad
	[d,e_i] = \delta_{i,0}e_i,\quad
	[d,f_i] = -\delta_{i,0}f_i\quad (i,j=0,\ldots,N),
\end{align*}
where $[\cdot,\cdot]:\mathfrak{g}\times\mathfrak{g}\to\mathfrak{g}$ is the Lie bracket.
We assume that the indices of $\alpha^{\vee}_i,e_i,f_i$ are congruent modulo $N+1$.
The canonical central element is given by
\[
	K = \sum_{i=0}^{N}\alpha^{\vee}_i.
\]
Note that
\[
	[\alpha^{\vee}_i,K] = 0,\quad
	[e_i,K] = 0,\quad
	[f_i,K] = 0,\quad
	[d,K] = 0\quad (i=0,\ldots,N).
\]
The normalized invariant form $(\cdot|\cdot):\mathfrak{g}\times\mathfrak{g}\to\mathbb{C}$ is defined by
\begin{align*}
	&(\alpha^{\vee}_i|\alpha^{\vee}_j) = a_{i,j},\quad
	(e_i|f_j) = \delta_{i,j},\quad
	(\alpha^{\vee}_i|e_j) = 0,\quad
	(\alpha^{\vee}_i|f_j) = 0, \\
	&(d|d) = 0,\quad
	(d|\alpha^{\vee}_i) = \delta_{0,i},\quad
	(d|e_i) = 0,\quad
	(d|f_i) = 0, \\
	&(K|K) = 0,\quad
	(K|\alpha^{\vee}_i) = 0,\quad
	(K|e_i) = 0,\quad
	(K|f_i) = 0,\quad
	(d|K) = 1\quad (i,j=0,\ldots,N).
\end{align*}
Note that the Lie bracket and the normalized invariant form satisfy the relation
\[
	(x|[y,z]) = ([x,y]|z).
\]

The Cartan subalgebra of $\mathfrak{g}$ is given by
\[
	\mathfrak{h} = \bigoplus_{i=0}^{N}\mathbb{C}\alpha^{\vee}_i \oplus \mathbb{C}d.
\]
The Borel subalgebras of $\mathfrak{g}$ are given by
\[
	\mathfrak{b}_{\pm} = \mathfrak{n}_{\pm}\oplus\mathfrak{h},
\]
where $\mathfrak{n}_{+}$ and $\mathfrak{n}_{-}$ are the subalgebras of $\mathfrak{g}$ generated by $e_0,\ldots,e_N$ and $f_0,\ldots,f_N$, respectively.
Let
\[
	\Lambda_i = \sum_{j=0}^{n}\mathrm{ad}\,e_{i+j(m+1)}\,\mathrm{ad}\,e_{i+1+j(m+1)}\,\ldots\,\mathrm{ad}\,e_{i+m-1+j(m+1)}(e_{i+m+j(m+1)})\quad (i=1,\ldots,m+1).
\]
We define the Heisenberg subalgebra of $\mathfrak{g}$ corresponding to the partition $(n+1,\ldots,n+1)$ of $N+1$ by
\[
	\mathfrak{s} = \left\{x\in\mathfrak{g}\bigm|[x,\Lambda_i]=0\ (i=1,\ldots,m+1)\right\}.
\]
Note that
\[
	[\Lambda_i,\Lambda_j] = 0\quad (i,j=1,\ldots,m+1).
\]

Let
\[
	\mathbf{s} = (s_0,\ldots,s_N),\quad
	s_i = \left\{\begin{array}{cc}1&(i\in(m+1)\mathbb{Z})\\[4pt]0&(i\notin(m+1)\mathbb{Z})\end{array}\right..
\]
Then we can take an element $\vartheta\in\mathfrak{h}$ such that
\[
	(\vartheta|\alpha^{\vee}_i) = s_i\quad (i=0,\ldots,N).
\]
For the explicit formula of $\vartheta$, see \cite{FS4}.
The $\mathbb{Z}$-gradation of type $\mathbf{s}$ is defined by
\[
	\mathfrak{g} = \bigoplus_{k\in\mathbb{Z}}\mathfrak{g}_k,\quad
	\mathfrak{g}_k = \left\{x\in\mathfrak{g}\bigm|[\vartheta,x]=kx\right\}.
\]
Note that
\begin{align*}
	&\deg\mathfrak{h} = \deg e_i = \deg f_i = 0\quad (i\notin(m+1)\mathbb{Z}), \\
	&\deg e_i = 1,\quad \deg f_i = -1\quad (i\in(m+1)\mathbb{Z}).
\end{align*}
Since
\[
	[\vartheta,\Lambda_i] = \Lambda_i\quad (i,j=1,\ldots,m+1),
\]
the Heisenberg subalgebra $\mathfrak{s}$ is compatible with this gradation, decomposing as
\[
	\mathfrak{s} = \bigoplus_{k\in\mathbb{Z}}\mathfrak{s}_k,\quad
	\mathfrak{s}_k = \left\{x\in\mathfrak{s}\bigm|[\vartheta,x]=kx\right\}.
\]
Note that
\[
	\deg K = 0,\quad
	\deg\Lambda_i = 1\quad (i=1,\ldots,m+1).
\]

The affine Lie algebra $\mathfrak{g}$ can be identified with a central extension of the loop algebra $\widehat{\mathfrak{sl}}_{N+1}=\mathfrak{sl}_{N+1}[z,z^{-1}]\oplus\mathbb{C}z\frac{d}{dz}\oplus\mathbb{C}K$.
The Chevalley generators and the scaling element of $\widehat{\mathfrak{sl}}_{N+1}$ are given by
\begin{align*}
	&\alpha^{\vee}_0 = K - E_{1,1} + E_{N+1,N+1},\quad
	e_0 = zE_{N+1,1},\quad
	f_0 = z^{-1}E_{1,N+1}, \\
	&\alpha^{\vee}_i = E_{i,i} - E_{i+1,i+1},\quad
	e_i = E_{i,i+1},\quad
	f_i = E_{i+1,i}\quad (i=1,\ldots,N),\quad
	d = z\frac{d}{dz},
\end{align*}
where $E_{i,j}$ is an $(N+1)\times(N+1)$ matrix with $1$ in the $(i,j)$-th entry and zeros elsewhere.
The Lie bracket for $\widehat{\mathfrak{sl}}_{N+1}$ is given by
\begin{align*}
	&[z^kX,z^lY] = z^{k+l}(XY-YX) + \delta_{k+l,0}\,k\,\mathrm{tr}(XY)K,\quad
	[K,z^kX] = 0, \\
	&[z\frac{d}{dz},z^kX] = kz^kX,\quad
	[z\frac{d}{dz},K] = 0\quad (X,Y\in\mathfrak{gl}_{N+1}).
\end{align*}
The normalized invariant form for $\widehat{\mathfrak{sl}}_{N+1}$ is given by
\begin{align*}
	&(z^kX|z^lY) = \delta_{k+l,0}\,\mathrm{tr}(XY),\quad
	(K|z^kX) = 0,\quad
	(K|K) = 0, \\
	&(z\frac{d}{dz}|z^kX) = 0,\quad
	(z\frac{d}{dz}|K) = 1,\quad
	(z\frac{d}{dz}|z\frac{d}{dz}) = 0\quad (X,Y\in\mathfrak{sl}_{N+1}).
\end{align*}

\section{Drinfeld-Sokolov hierarchy of type $A$}\label{Sec:DS}

In this section we formulate the DS hierarchy and its similarity reduction on the subalgebra $\mathfrak{g}_0$.
We also transform the similarity reduction to that on the Borel subalgebra $\mathfrak{b}_{+}$ by a gauge transformation.

\subsection{DS hierarchy}

Let $t_i\ (i=1,\ldots,m+1)$ be independent variables corresponding to the generators $\Lambda_i\in\mathfrak{s}$.
We consider a function
\[
	W = \exp\left(\sum_{k=1}^{\infty}w_k\right),\quad
	w_k \in \mathfrak{g}_{-k}.
\]
In the following, we use the convention
\[
	e^xye^{-x} = \exp(\mathrm{ad}\,x)(y) = \sum_{k=0}^{\infty}\frac{1}{k!}(\mathrm{ad}\,x)^k(y)\quad (x,y\in\mathfrak{g}).
\]
Let $\partial_i\ (i=1,\ldots,m+1)$ be partial differential operators with respect to $t_i$.
Assume that the Sato equations
\begin{equation}\label{Eq:Sato}
	U_i = \partial_i - \Lambda_i - W(\partial_i-\Lambda_i)W^{-1} \in \mathfrak{g}_0\quad (i=1,\ldots,m+1)
\end{equation}
hold.
Then we obtain the Zakharov-Shabat equations
\begin{equation}\label{Eq:ZS}
	[\partial_i-U_i-\Lambda_i,\partial_j-U_j-\Lambda_j] = 0\quad (i,j=1,\ldots,m+1),
\end{equation}
or equivalently,
\begin{equation}\label{Eq:DS1}
	\partial_i(U_j) - \partial_j(U_i) + [U_j,U_i] = 0,\quad
	[U_i,\Lambda_j] - [U_j,\Lambda_i] = 0\quad (i,j=1,\ldots,m+1).
\end{equation}
Note that system \eqref{Eq:ZS} is derived via a gauge transformation
\[
	W[\partial_i-\Lambda_i,\partial_j-\Lambda_j]W^{-1} = 0.
\]

\begin{lem}
We have the relations
\begin{equation}\label{Eq:DS2}
	2(\vartheta|\partial_i(U_j)) + (U_i|U_j) = 0\quad (i,j=1,\ldots,m+1).
\end{equation}
\end{lem}

\begin{proof}
System \eqref{Eq:Sato} can be described explicitly as
\[
	U_i = \mathrm{ad}\,w_1(\Lambda_i)
\]
in degree $0$ and
\begin{equation}\label{Eq:Lem_DS2_Proof}
	\mathrm{ad}\,w_2(\Lambda_i) + \frac12(\mathrm{ad}\,w_1)^2(\Lambda_i) + \partial_i(w_1) = 0
\end{equation}
in degree $(-1)$.
It follows that
\begin{align*}
	(\Lambda_j|\text{LHS of \eqref{Eq:Lem_DS2_Proof}})
	&= (\Lambda_j|[w_2,\Lambda_i]) + \frac12(\Lambda_j|[w_1,[w_1,\Lambda_i]]) + ([\vartheta,\Lambda_j]|\partial_i(w_1)) \\
	&= ([\Lambda_i,\Lambda_j]|w_2) - \frac12([w_1,\Lambda_j]|[w_1,\Lambda_i]) - (\vartheta|[\partial_i(w_1),\Lambda_j]) \\
	&= 0 - \frac12(U_j|U_i) - (\vartheta|\partial_i(U_j)),
\end{align*}
from which we obtain system \eqref{Eq:DS2}.
\end{proof}

We call systems \eqref{Eq:DS1} and \eqref{Eq:DS2} the DS hierarchy in this article.

\begin{rem}
In general, the DS hierarchy is formulated in terms of all generators of $\mathfrak{s}$.
In this article, we use only degree 1 generators $\Lambda_i$ to derive a Painlev\'e system of the lowest order.
\end{rem}

\subsection{Similarity reduction on $\mathfrak{g}_0$}\label{Subsec:SR_g0}

Let
\[
	K_i = \sum_{j=0}^{n}\alpha^{\vee}_{i+j(m+1)}\quad (i=1,\ldots,m).
\]
Then they satisfy
\[
	[K_i,K_j] = 0,\quad
	[K_i,\Lambda_k] = 0,\quad
	[\vartheta,K_i] = 0,\quad
	(\vartheta|K_i) = 0\quad (i,j=1,\ldots,m,\ k=1,\ldots,m+1).
\]
Note that
\begin{equation}\label{Eq:DS3}
	(K_i|U_j) = (K_i|[w_1,\Lambda_j])
	= ([\Lambda_j,K_i]|w_1)
	= 0\quad (i=1,\ldots,m,\ j=1,\ldots,m+1).
\end{equation}
Introducing new constants $\rho_i\ (i=1,\ldots,m)$, we set
\[
	\rho = \sum_{i=1}^{m}\rho_iK_i.
\]
We also set
\[
	\widehat{U} = \rho + \sum_{i=1}^{m+1}t_iU_i,\quad
	\widehat{\Lambda} = \sum_{j=1}^{m+1}t_j\Lambda_j.	
\]
Assume that the Sato-like equation
\[
	\vartheta - \widehat{U} - \widehat{\Lambda} = W\left(\vartheta-\rho-\widehat{\Lambda}\right)W^{-1}
\]
holds.
Then we obtain the Zakharov-Shabat equations
\begin{equation}\label{Eq:ZS_SR}
	[\partial_i-U_i-\Lambda_i,\vartheta-\widehat{U}-\widehat{\Lambda}] = 0\quad (i=1,\ldots,m+1),
\end{equation}
or equivalently,
\begin{equation}\begin{split}\label{Eq:DS_SR1}
	\partial_i(\widehat{U}) - [U_i,\widehat{U}] = 0,\quad
	[U_i,\widehat{\Lambda}] - [\widehat{U},\Lambda_i] = 0\quad (i=1,\ldots,m+1).
\end{split}\end{equation}
Note that system \eqref{Eq:ZS_SR} is derived via a gauge transformation
\[
	W[\partial_i-\Lambda_i,\vartheta-\rho-\widehat{\Lambda}]W^{-1} = 0.
\]

\begin{lem}
We have the relations
\begin{equation}\label{Eq:DS_SR2}
	2(\vartheta|U_i) - (\widehat{U}|U_i) = 0\quad (i=1,\ldots,m+1).
\end{equation}
\end{lem}

\begin{proof}
We can rewrite the first equation of system \eqref{Eq:DS_SR1} as
\begin{align*}
	\partial_i(\widehat{U}) - [U_i,\widehat{U}] &= \partial_i\left(\rho+\sum_{j=1}^{m+1}t_jU_j\right) - \left[U_i,\rho+\sum_{j=1}^{m+1}t_jU_j\right] \\
	&= \sum_{j=1}^{m+1}t_j\partial_i(U_j) + U_i + [\rho,U_i] + \sum_{j=1}^{m+1}t_j[U_j,U_i] \\
	&= 0.
\end{align*}
by using system \eqref{Eq:DS1}.
Furthermore, relation \eqref{Eq:DS3} implies that $(\rho|U_i)=0$.
Then relation \eqref{Eq:DS2} implies
\begin{align*}
	2\sum_{j=1}^{m+1}t_j(\vartheta|\partial_i(U_j)) + \sum_{j=1}^{m+1}t_j(U_i|U_j) &= -2(\vartheta|U_i) - 2(\vartheta|[\rho,U_i]) - 2\sum_{j=1}^{m+1}t_j(\vartheta|[U_j,U_i]) + (U_i|\widehat{U}-\rho) \\
	&= -2(\vartheta|U_i) - 2([\vartheta,\rho]|U_i) - 2\sum_{j=1}^{m+1}t_j([\vartheta,U_j]|U_i) + (\widehat{U}|U_i) - (\rho|U_i) \\
	&= -2(\vartheta|U_i) + (\widehat{U}|U_i) \\
	&= 0.
\end{align*}
Hence, we have proved the lemma.
\end{proof}

We call systems \eqref{Eq:DS_SR1} and \eqref{Eq:DS_SR2} the similarity reduction of the DS hierarchy.
Since the relation
\[
	\sum_{i=1}^{m+1}t_i\partial_i(\widehat{U}) + [\rho,\widehat{U}] = 0
\]
holds, the similarity reduction is regarded as a system of partial differential equations in $m$ variables.
Note that relations \eqref{Eq:DS_SR2} determine the components of the canonical central element $K$ in the functions $U_i$.

\subsection{Similarity reduction on $\mathfrak{b}_{+}$}

We consider a function 
\[
	W_0 = \exp(w_0),\quad
	w_0 \in \mathfrak{g}_0\cap\mathfrak{n}_{-}
\]
satisfying the relation
\begin{equation}\label{Eq:Gauge_Borel}
	\mathcal{M} = \vartheta - W_0(\vartheta-\widehat{U}-\widehat{\Lambda})W_0^{-1} \in \mathfrak{b}_{+}.
\end{equation}
We also set
\[
	\mathcal{B}_i = \partial_i - W_0(\partial_i-U_i-\Lambda_i)W_0^{-1}\quad (i=1,\ldots,m+1).
\]
Then the Zakharov-Shabat equations
\begin{equation}\label{Eq:DS_SR_Borel1}
	[\partial_i-\mathcal{B}_i,\vartheta-\mathcal{M}] = 0\quad (i=1\ldots,m+1)
\end{equation}
are derived via a gauge transformation
\[
	W_0[\partial_i-U_i-\Lambda_i,\vartheta-\widehat{U}-\widehat{\Lambda}]W_0^{-1} = 0.
\]
We denote the $\mathfrak{h}$-component of the function $\mathcal{M}$ by
\[
	\mathcal{M}|_{\mathfrak{h}} = \sum_{i=0}^{N}\kappa_i\alpha^{\vee}_i.
\]

\begin{prop}
If we assume that
\[
	-\kappa_{k-1+j(m+1)}+\kappa_{k+j(m+1)}+\kappa_{l+j(m+1)}-\kappa_{l+1+j(m+1)} \neq 0\quad (j=0,\ldots,n,\ k,l=1,\ldots,m,\ k\leq l),
\]
then the relations $\mathcal{B}_i\in\mathfrak{b}_{+}\ (i=1,\ldots,m+1)$ hold.
Moreover, we have the equations
\[
	\partial_i(\kappa_j) = 0\quad (i=1,\ldots,m+1,\ j=0,\ldots,N).
\]
\end{prop}

\begin{proof}
We temporarily use the notations
\begin{align*}
	e^{j}_{k,l} &= \mathrm{ad}\,e_{k+j(m+1)}\,\mathrm{ad}\,e_{k+1+j(m+1)}\,\ldots\,\mathrm{ad}\,e_{l-1+j(m+1)}(e_{l+j(m+1)}), \\
	f^{j}_{l,k} &= \mathrm{ad}\,f_{l+j(m+1)}\,\mathrm{ad}\,f_{l-1+j(m+1)}\,\ldots\,\mathrm{ad}\,f_{k+1+j(m+1)}(f_{k+j(m+1)}), \\
	u^{j}_{l,k} &= (\partial_i(W_0)W_0^{-1}+W_0U_iW_0^{-1}|e^{j}_{k,l})\quad (j=0,\ldots,n,\ k,l=1,\ldots,m,\ k\leq l).
\end{align*}
The $\mathfrak{b}_+$-valued functions $\mathcal{M}$ and $\mathcal{B}_i$ are described as
\[
	\mathcal{M} = W_0\widehat{U}W_0^{-1} + W_0\widehat{\Lambda}W_0^{-1},\quad
	\mathcal{B}_i = \partial_i(W_0)W_0^{-1} + W_0U_iW_0^{-1} + W_0\Lambda_iW_0^{-1}.
\]
Hence, the $\mathfrak{g}_0$-component of system \eqref{Eq:DS_SR_Borel1} is described as
\[
	\partial_i(W_0\widehat{U}W_0^{-1}) = [\partial_i(W_0)W_0^{-1}+W_0U_iW_0^{-1},W_0\widehat{U}W_0^{-1}].
\]
Then we have
\begin{align*}
	(\partial_i(W_0\widehat{U}W_0^{-1})|e^{j}_{1,m}) &= (\partial_i(W_0)W_0^{-1}+W_0U_iW_0^{-1}|[W_0\widehat{U}W_0^{-1},e^{j}_{1,m}]) \\
	&= (u^{j}_{m,1}f^{j}_{m,1}|(-\kappa_{j(m+1)}+\kappa_{1+j(m+1)}+\kappa_{m+j(m+1)}-\kappa_{(j+1)(m+1)})e^{j}_{1,m}) \\
	&= u^{j}_{m,1}(-\kappa_{j(m+1)}+\kappa_{1+j(m+1)}+\kappa_{m+j(m+1)}-\kappa_{(j+1)(m+1)}),
\end{align*}
where $\kappa_{N+1}=\kappa_0$.
Note that relation \eqref{Eq:Gauge_Borel} implies $(W_0\widehat{U}W_0^{-1}|e^{j}_{1,m})=0$.
The assumption implies that $u^{j}_{m,1}=0$.
Furthermore, we have
\begin{align*}
	(\partial_i(W_0\widehat{U}W_0^{-1})|e^{j}_{1,m-1}) &= u^{j}_{m-1,1}(-\kappa_{j(m+1)}+\kappa_{1+j(m+1)}+\kappa_{m-1+j(m+1)}-\kappa_{m+j(m+1)}), \\
	(\partial_i(W_0\widehat{U}W_0^{-1})|e^{j}_{2,m}) &= u^{j}_{m,2}(-\kappa_{1+j(m+1)}+\kappa_{2+j(m+1)}+\kappa_{m+j(m+1)}-\kappa_{(j+1)(m+1)}).
\end{align*}
The assumption also implies that $u^{j}_{m-1,1}=0$ and $u^{j}_{m,2}=0$.
In a similar manner, we obtain
\[
	u^{j}_{k,l} = 0\quad (j=0,\ldots,n,\ k,l=1,\ldots,m,\ k\leq l).
\]
Thus, we have proved the first half of the proposition.
Moreover, the relation
\[
	\partial_i(\mathcal{M}|_{\mathfrak{h}}) = [\mathcal{B}_i,\mathcal{M}]|_{\mathfrak{h}}\quad (i=1,\ldots,m+1)
\]
and the first half prove the latter half.
\end{proof}

\begin{lem}
We have the relations
\begin{equation}\label{Eq:DS_SR_Borel2}
	2(\vartheta|\mathcal{B}_i) - (\mathcal{M}|\mathcal{B}_i-\partial_i(W_0)W_0^{-1}) = 0\quad (i=1,\ldots,m+1).
\end{equation}
\end{lem}

\begin{proof}
We apply the gauge transformation to relation \eqref{Eq:DS_SR2}.
Then the first term of the left-hand side is transformed as
\begin{align*}
	(\vartheta|U_i) &= (W_0\vartheta W_0^{-1}|W_0U_iW_0^{-1}) \\
	&= (\vartheta|\mathcal{B}_i-\partial_i(W_0)W_0^{-1}-W_0\Lambda_iW_0^{-1}) \\
	&= (\vartheta|\mathcal{B}_i).
\end{align*}
Similarly, the second term is transformed as
\begin{align*}
	(\widehat{U}|U_i) &= (W_0\widehat{U}W_0^{-1}|W_0U_iW_0^{-1}) \\
	&= (\mathcal{M}-W_0\widehat{\Lambda}W_0^{-1}|\mathcal{B}_i-\partial_i(W_0)W_0^{-1}-W_0\Lambda_iW_0^{-1}) \\
	&= (\mathcal{M}|\mathcal{B}_i-\partial_i(W_0)W_0^{-1}).
\end{align*}
Hence, we have proved the lemma.
\end{proof}

We call systems \eqref{Eq:DS_SR_Borel1} and \eqref{Eq:DS_SR_Borel2} the similarity reduction on $\mathfrak{b}_+$.

\section{Hamiltonian system}\label{Sec:Ham}

In this section we express the similarity reductions given in the previous section as Hamiltonian systems.
Ultimately, we arrive at a Hamiltonian system expressed in terms of canonical coordinates.
In the following, we regard the loop algebra $\mathfrak{sl}_{N+1}[z,z^{-1}]$ as the affine Lie algebra $\mathfrak{g}$.

\subsection{Hamiltonian system on $\mathfrak{g}_0$}

Let
\[
	E^{j}_{k,l} = E_{k+j(m+1),l+j(m+1)}\quad (j=0,\ldots,n,\ k,l=1,\ldots,m+1).
\]
We denote the set of these matrices by $\mathfrak{e}_0$.
The subalgebra $\mathfrak{g}_0$ is generated by the matrices $E^{j}_{k,l}\ (k\neq l)$ and $E^{j_1}_{k_1,k_1}-E^{j_2}_{k_2,k_2}$.

\begin{lem}\label{Lem:DS_SR_Uhat_Ui}
The $\mathfrak{g}_0$-valued functions $U_i\ (i=1,\ldots,m+1)$ are described in terms of $\widehat{U}$ as
\begin{align*}
	&(E^{j}_{i,k}|\widehat{U}) = t_i(E^{j}_{i,k}|U_i) - t_k(E^{j+1}_{i,k}|U_i), \\
	&(E^{j}_{k,i}|\widehat{U}) = t_i(E^{j}_{k,i}|U_i) - t_k(E^{j-1}_{k,i}|U_i)\quad (j\in\mathbb{Z}/(n+1)\mathbb{Z},\ k=1,\ldots,i-1), \\
	&(E^{j}_{k,i}|\widehat{U}) = t_i(E^{j}_{k,i}|U_i) - t_k(E^{j-1}_{k,i}|U_i), \\
	&(E^{j}_{i,k}|\widehat{U}) = t_i(E^{j}_{i,k}|U_i) - t_k(E^{j+1}_{i,k}|U_i)\quad (j\in\mathbb{Z}/(n+1)\mathbb{Z},\ k=i+1,\ldots,m+1)
\end{align*}
and
\[
	(E^{j}_{i,i}|U_i) = \frac{1}{t_i}(E^{j}_{i,i}|\widehat{U}-\rho),\quad
	(E^{j}_{k,k}|U_i) = 0\quad (j=0,\ldots,n,\ k=1,\ldots,m+1,\ k\neq i).
\]
\end{lem}

\begin{proof}
We temporarily use the notations
\begin{align*}
	&\widetilde{E}^{j}_{k,l} = E_{k+j(m+1),l+(j+1)(m+1)}\quad (j=0,\ldots,n-1,\ k,l=1,\ldots,m+1), \\
	&\widetilde{E}^{n}_{k,l} = zE_{k+n(m+1),l}\quad (k,l=1,\ldots,m+1).
\end{align*}
The functions $U_i$ and the generators $\Lambda_i$ are described in terms of $E^{j}_{k,l}$ and $\widetilde{E}^{j}_{k,l}$ as
\begin{align*}
	U_i &= \sum_{j=0}^{n}(E^{j}_{i,i}|U_i)E^{j}_{i,i} + \sum_{j=0}^{n}\sum_{k=1}^{i-1}(E^{j}_{k,i}|U_i)E^{j}_{i,k} + \sum_{j=0}^{n}\sum_{k=i+1}^{m+1}(E^{j}_{k,i}|U_i)E^{j}_{i,k} \\
	&\quad + \sum_{j=0}^{n}\sum_{k=1}^{i-1}(E^{j}_{i,k}|U_i)E^{j}_{k,i} + \sum_{j=0}^{n}\sum_{k=i+1}^{m+1}(E^{j}_{i,k}|U_i)E^{j}_{k,i}
\end{align*}
and
\[
	\Lambda_i = \sum_{j=0}^{n}\widetilde{E}^{j}_{i,i},
\]
respectively.
Note that the relations
\[
	\mathrm{tr}\,U_i = \sum_{j=0}^{n}(E^{j}_{i,i}|U_i) = 0
\]
hold since $U_i$ are $\mathfrak{sl}_{N+1}$-valued functions.
Then we have
\begin{align*}
	[U_i,\widehat{\Lambda}] &= \sum_{j=0}^{n-1}t_i(E^{j}_{i,i}-E^{j+1}_{i,i}|U_i)\widetilde{E}^{j}_{i,i} + t_i(E^{n}_{i,i}-E^{0}_{i,i}|U_i)\widetilde{E}^{n}_{i,i} \\
	&= \sum_{j=0}^{n-1}\sum_{k=1}^{i-1}\left(t_k(E^{j}_{k,i}|U_i)-t_i(E^{j+1}_{k,i}|U_i)\right)\widetilde{E}^{j}_{i,k} + \sum_{k=1}^{i-1}\left(t_k(E^{n}_{k,i}|U_i)-t_i(E^{0}_{k,i}|U_i)\right)\widetilde{E}^{n}_{i,k} \\
	&\quad + \sum_{j=0}^{n-1}\sum_{k=i+1}^{m+1}\left(t_k(E^{j}_{k,i}|U_i)-t_i(E^{j+1}_{k,i}|U_i)\right)\widetilde{E}^{j}_{i,k} + \sum_{k=i+1}^{m+1}\left(t_k(E^{n}_{k,i}|U_i)-t_i(E^{0}_{k,i}|U_i)\right)\widetilde{E}^{n}_{i,k} \\
	&\quad + \sum_{j=0}^{n-1}\sum_{k=1}^{i-1}\left(t_i(E^{j}_{i,k}|U_i)-t_k(E^{j+1}_{i,k}|U_i)\right)\widetilde{E}^{j}_{k,i} + \sum_{k=1}^{i-1}\left(t_i(E^{n}_{i,k}|U_i)-t_k(E^{0}_{i,k}|U_i)\right)\widetilde{E}^{n}_{k,i} \\
	&\quad + \sum_{j=0}^{n-1}\sum_{k=i+1}^{m+1}\left(t_i(E^{j}_{i,k}|U_i)-t_k(E^{j+1}_{i,k}|U_i)\right)\widetilde{E}^{j}_{k,i} + \sum_{k=i+1}^{m+1}\left(t_i(E^{n}_{i,k}|U_i)-t_k(E^{0}_{i,k}|U_i)\right)\widetilde{E}^{n}_{k,i}
\end{align*}
and
\begin{align*}
	[\widehat{U},\Lambda_i] &= \sum_{j=0}^{n}\sum_{k=1}^{m+1}(E^{j}_{i,k}|\widehat{U})\widetilde{E}^{j}_{k,i} - \sum_{j=0}^{n-1}\sum_{k=1}^{m+1}(E^{j+1}_{k,i}|\widehat{U})\widetilde{E}^{j}_{i,k} - \sum_{k=1}^{m+1}(E^{0}_{k,i}|\widehat{U})\widetilde{E}^{n}_{i,k}.
\end{align*}
Hence, we can prove the first half of the lemma by using the second equation of system \eqref{Eq:DS_SR1}.
Moreover, the relation
\[
	\sum_{i=1}^{m+1}t_iU_i = \widehat{U} - \rho
\]
proves the latter half.
\end{proof}

We define a Lie Poisson bracket for the $\mathfrak{g}_0$-valued function $\widehat{U}$ by
\begin{equation}\label{Eq:Poisson}
	\{(x|\widehat{U}),(y|\widehat{U})\} = (n+1)([y,x]|\widehat{U})\quad (x,y\in\mathfrak{e}_0).
\end{equation}
We also assume that
\[
	\{\rho_i,(x|\widehat{U})\} = 0\quad (i=1,\ldots,m,\ x\in\mathfrak{e}_0).
\]

\begin{thm}\label{Thm:DS_SR_Ham}
System \eqref{Eq:DS_SR1} is expressed as a Hamiltonian system
\begin{equation}\label{Eq:DS_SR_Ham}
	\partial_i((x|\widehat{U})) = \{H_i,(x|\widehat{U})\},\quad
	H_i = \frac{(U_i|\widehat{U})}{2(n+1)}\quad (i=1,\ldots,m+1,\ x\in\mathfrak{e}_0).
\end{equation}
\end{thm}

\begin{proof}
We temporarily use the notation
\[
	\overline{E^{j}_{k,l}} = E^{j}_{l,k}\quad (j=0,\ldots,n,\ k,l=1,\ldots,m+1).
\]
The left-hand side of system \eqref{Eq:DS_SR_Ham} is rewritten as
\begin{align*}
	\partial_i((x|\widehat{U})) &= (x|\partial_i(\widehat{U})) \\
	&= (x|[U_i,\widehat{U}]) \\
	&= \sum_{y\in\mathfrak{e}_0}(x|(\overline{y}|U_i)[y,\widehat{U}]) \\
	&= \sum_{y\in\mathfrak{e}_0}(\overline{y}|U_i)([x,y]|\widehat{U}) \\
	&= \frac{1}{n+1}\sum_{y\in\mathfrak{e}_0}(\overline{y}|U_i)\{(y|\widehat{U}),(x|\widehat{U})\}.
\end{align*}
The right-hand side of system \eqref{Eq:DS_SR_Ham} is rewritten as
\begin{align*}
	\{H_i,(x|\widehat{U})\} &= \frac{1}{2(n+1)}\{(U_i|\widehat{U}),(x|\widehat{U})\} \\
	&= \frac{1}{2(n+1)}\sum_{y\in\mathfrak{e}_0}\{(\overline{y}|U_i)(y|\widehat{U}),(x|\widehat{U})\} \\
	&= \frac{1}{2(n+1)}\sum_{y\in\mathfrak{e}_0}\left((\overline{y}|U_i)\{(y|\widehat{U}),(x|\widehat{U})\}+(y|\widehat{U})\{(\overline{y}|U_i),(x|\widehat{U})\}\right).
\end{align*}
Hence, it is enough to verify the relation
\begin{equation}\label{Eq:DS_SR_Ham_Proof_1}
	\sum_{y\in\mathfrak{e}_0}(\overline{y}|U_i)\{(y|\widehat{U}),(x|\widehat{U})\} = \sum_{y\in\mathfrak{e}_0}(y|\widehat{U})\{(\overline{y}|U_i),(x|\widehat{U})\}\quad (i=1,\ldots,m+1,\ x\in\mathfrak{e}_0).
\end{equation}
The left-hand side of relation \eqref{Eq:DS_SR_Ham_Proof_1} is described in terms of the elements of $\mathfrak{e}_0$ as
\begin{equation}\label{Eq:DS_SR_Ham_Proof_2}\begin{split}
	\sum_{y\in\mathfrak{e}_0}(\overline{y}|U_i)\{(y|\widehat{U}),(x|\widehat{U})\} &= \sum_{j=0}^{n}(E^{j}_{i,i}|U_i)\{(E^{j}_{i,i}|\widehat{U}),(x|\widehat{U})\} \\
	&\quad + \sum_{j=0}^{n}\sum_{k=1}^{i-1}(E^{j}_{k,i}|U_i)\{(E^{j}_{i,k}|\widehat{U}),(x|\widehat{U})\} + \sum_{j=0}^{n}\sum_{k=i+1}^{m+1}(E^{j}_{k,i}|U_i)\{(E^{j}_{i,k}|\widehat{U}),(x|\widehat{U})\} \\
	&\quad + \sum_{j=0}^{n}\sum_{k=1}^{i-1}(E^{j}_{i,k}|U_i)\{(E^{j}_{k,i}|\widehat{U}),(x|\widehat{U})\} + \sum_{j=0}^{n}\sum_{k=i+1}^{m+1}(E^{j}_{i,k}|U_i)\{(E^{j}_{k,i}|\widehat{U}),(x|\widehat{U})\}.
\end{split}\end{equation}
We can rewrite the first term of the right-hand side of relation \eqref{Eq:DS_SR_Ham_Proof_2} as
\begin{align*}
	\sum_{j=0}^{n}(E^{j}_{i,i}|U_i)\{(E^{j}_{i,i}|\widehat{U}),(x|\widehat{U})\} &= \sum_{j=0}^{n}(E^{j}_{i,i}|U_i)\{(E^{j}_{i,i}|\widehat{U}-\rho),(x|\widehat{U})\} \\
	&= \sum_{j=0}^{n}(E^{j}_{i,i}|\widehat{U}-\rho)\{(E^{j}_{i,i}|U_i),(x|\widehat{U})\} \\
	&= \sum_{j=0}^{n}(E^{j}_{i,i}|\widehat{U})\{(E^{j}_{i,i}|U_i),(x|\widehat{U})\} - \{(\rho|U_i),(x|\widehat{U})\} \\
	&= \sum_{j=0}^{n}(E^{j}_{i,i}|\widehat{U})\{(E^{j}_{i,i}|U_i),(x|\widehat{U})\}.
\end{align*}
by using relation \eqref{Eq:DS3} and Lemma \ref{Lem:DS_SR_Uhat_Ui}.
The second term is rewritten as
\begin{align*}
	\sum_{j=0}^{n}\sum_{k=1}^{i-1}(E^{j}_{k,i}|U_i)\{(E^{j}_{i,k}|\widehat{U}),(x|\widehat{U})\} &= \sum_{j=0}^{n}\sum_{k=1}^{i-1}(E^{j}_{k,i}|U_i)\left(t_i\{(E^{j}_{i,k}|U_i),(x|\widehat{U})\}-t_k\{(E^{j+1}_{i,k}|U_i),(x|\widehat{U})\}\right) \\
	&= \sum_{j=0}^{n}\sum_{k=1}^{i-1}\left(t_i(E^{j}_{k,i}|U_i)-t_k(E^{j-1}_{k,i}|U_i)\right)\{(E^{j}_{i,k}|U_i),(x|\widehat{U})\} \\
	&= \sum_{j=0}^{n}\sum_{k=1}^{i-1}(E^{j}_{k,i}|\widehat{U})\{(E^{j}_{i,k}|U_i),(x|\widehat{U})\}.
\end{align*}
Here we assume that the index $j$ is congruent modulo $n+1$.
The third, fourth and fifth terms can be rewritten in a similar manner.
Hence, we obtain relation \eqref{Eq:DS_SR_Ham_Proof_1}.
\end{proof}

\begin{rem}
The similarity reduction can be regarded as a Schlesinger system of irregular singular type via a suitable gauge transformation.
Moreover, the Hamiltonian structure is well known for the Schlesinger system; see \cite{Sak,Suz2} and references therein.
Therefore, the result of Theorem \ref{Thm:DS_SR_Ham} should be equivalent to this known structure.
However, using this known result requires identifying the gauge freedom.
Since this identification is technically involved, we provide a direct proof using the normalized invariant form.
On the other hand, a Poisson bracket for the DS hierarchy was proposed in \cite{BGHM}, whose relationship with \eqref{Eq:Poisson} has not been clarified yet.
\end{rem}

\subsection{Hamiltonian system on $\mathfrak{b}_{+}$}

Let
\[
	\varphi^{j}_{k,l} = (E^{j}_{l,k}|\mathcal{M}),\quad
	w^{j}_{l,k} = (E^{j}_{k,l}|W_0)\quad (j=0,\ldots,n,\ k,l=1,\ldots,m+1,\ k<l).
\]
We define a Poisson bracket for these functions by
\begin{equation}\begin{split}\label{Eq:Poisson_Borel}
	&\{\varphi^{j_1}_{k_1,l_1},w^{j_2}_{l_2,k_2}\} = \left\{\begin{array}{cc}-(n+1)\delta_{j_1,j_2}\delta_{l_1,l_2}w_{k_1,k_2}&(k_1=k_2,\ldots,l_1-1)\\[4pt]0&(\text{otherwise})\end{array}\right., \\
	&\{\varphi^{j_1}_{k_1,l_1},\varphi^{j_2}_{k_2,l_2}\} = -(n+1)\delta_{j_1,j_2}(\delta_{l_1,k_2}\varphi^{j_1}_{k_1,l_2}-\delta_{l_2,k_1}\varphi^{j_1}_{k_2,l_1}), \\
	&\{w^{j_1}_{l_1,k_1},w^{j_2}_{l_2,k_2}\} = 0,
\end{split}\end{equation}
where $\delta_{k,i}$ is the Kronecker delta.
We also assume that
\[
	\{\kappa_i,\varphi^{j}_{k,l}\} = 0\quad
	\{\kappa_i,w^{j}_{l,k}\} = 0\quad (i=0,\ldots,N).
\]

Let
\[
	\overline{w}^{j}_{l,k} = (E^{j}_{k,l}|W_0^{-1})\quad (j=0,\ldots,n,\ k,l=1,\ldots,m+1,\ k<l).
\]
Since $\det W_0=1$, each $\overline{w}^{j}_{l,k}$ is a polynomial in $(w^{j}_{l,k})$.
The following lemma is used in the proof of the main theorem.

\begin{lem}
The relations
\[
	\{\varphi^{j_1}_{k_1,l_1},\overline{w}^{j_2}_{l_2,k_2}\} = \left\{\begin{array}{cc}(n+1)\delta_{j_1,j_2}\delta_{k_1,k_2}\overline{w}_{l_2,l_1}&(l_1=k_1+1,\ldots,l_2)\\[4pt]0&(\text{otherwise})\end{array}\right.
\]
hold.
\end{lem}

\begin{proof}
Since it is enough to verify the relations for the case $j_1=j_2$, we omit the index $j$.
We have
\[
	\sum_{r=1}^{m+1}w_{l,r}\overline{w}_{r,k} = \sum_{r=k}^{l}w_{l,r}\overline{w}_{r,k} = \delta_{k,l}.
\]
Hence, it is enough to verify that
\[
	\sum_{r=k_2}^{l_2}\{\varphi_{k_1,l_1},w_{l_2,r}\overline{w}_{r,k_2}\} = 0
\]
for $k_1<l_1$ and $k_2<l_2$.
Its left-hand side is rewritten as
\begin{equation}\label{Eq:Poisson_Borel_Lem_Pf_1}\begin{split}
	&\sum_{r=k_2}^{l_2}\{\varphi_{k_1,l_1},w_{l_2,r}\overline{w}_{r,k_2}\} \\
	&= \{\varphi_{k_1,l_1},w_{l_2,k_2}\} + \sum_{r=k_2+1}^{l_2-1}\{\varphi_{k_1,l_1},w_{l_2,r}\overline{w}_{r,k_2}\} + \{\varphi_{k_1,l_1},\overline{w}_{l_2,k_2}\} \\
	&= \{\varphi_{k_1,l_1},w_{l_2,k_2}\} + \sum_{r=k_2+1}^{l_2-1}\overline{w}_{r,k_2}\{\varphi_{k_1,l_1},w_{l_2,r}\} + \sum_{r=k_2+1}^{l_2-1}w_{l_2,r}\{\varphi_{k_1,l_1},\overline{w}_{r,k_2}\} + \{\varphi_{k_1,l_1},\overline{w}_{l_2,k_2}\}.
\end{split}\end{equation}
Then the first and second terms of the right-hand side of \eqref{Eq:Poisson_Borel_Lem_Pf_1} are rewritten as
\begin{align*}
	\{\varphi_{k_1,l_1},w_{l_2,k_2}\} + \sum_{r=k_2+1}^{l_2-1}\overline{w}_{r,k_2}\{\varphi_{k_1,l_1},w_{l_2,r}\} &= -(n+1)\delta_{l_1,l_2}\left(w_{k_1,k_2}+\sum_{r=k_2+1}^{\min\{k_1,l_2-1\}}w_{k_1,r}\overline{w}_{r,k_2}\right) \\
	&= (n+1)\delta_{l_1,l_2}\sum_{r=l_2}^{k_1}w_{k_1,r}\overline{w}_{r,k_2} \\
	&= (n+1)\delta_{l_1,l_2}\sum_{r=l_1}^{k_1}w_{k_1,r}\overline{w}_{r,k_2}.
\end{align*}
The right-hand side turns out to be zero.
We can show that the third and fourth terms of the right-hand side of \eqref{Eq:Poisson_Borel_Lem_Pf_1} vanish in a similar manner.
Hence, we have proved the lemma.
\end{proof}

\begin{thm}\label{Thm:Poisson_Borel}
Poisson bracket \eqref{Eq:Poisson_Borel} implies \eqref{Eq:Poisson}.
Namely, if the functions $\varphi^{j}_{k,l},w^{j}_{l,k}$ satisfy relation \eqref{Eq:Poisson_Borel}, then the $\mathfrak{g}_0$-valued function $\widehat{U}=W_0^{-1}\mathcal{M}W_0-\widehat{\Lambda}$ satisfies relation \eqref{Eq:Poisson}.
\end{thm}

\begin{proof}
As in the proof of the previous lemma, we omit the index $j$ of $\varphi^{j}_{k,l}$, $w^{j}_{l,k}$ and $\overline{w}^{j}_{l,k}$.
It is enough to verify that
\[
	\{(E^{j}_{l_1,k_1}|\widehat{U}),(E^{j}_{l_2,k_2}|\widehat{U})\} = (n+1)([E^{j}_{l_2,k_2},E^{j}_{l_1,k_1}]|\widehat{U}).
\]
The left-hand side is rewritten as
\begin{equation}\label{Eq:Poisson_Borel_Thm_Pf_1}\begin{split}
	\{(E^{j}_{l_1,k_1}|\widehat{U}),(E^{j}_{l_2,k_2}|\widehat{U})\} &= \left\{(E^{j}_{l_1,k_1}|W_0^{-1}\mathcal{M}W_0),(E^{j}_{l_2,k_2}|W_0^{-1}\mathcal{M}W_0)\right\} \\
	&= \sum_{r_1=1}^{k_1}\sum_{s_1=l_1}^{m+1}\sum_{r_2=1}^{k_2}\sum_{s_2=l_2}^{m+1}\{\overline{w}_{k_1,r_1}\varphi_{r_1,s_1}w_{s_1,l_1},\overline{w}_{k_2,r_2}\varphi_{r_2,s_2}w_{s_2,l_2}\} \\
	&= \sum_{r_1=1}^{k_1}\sum_{s_1=l_1}^{m+1}\sum_{r_2=1}^{k_2}\sum_{s_2=l_2}^{m+1}\varphi_{r_1,s_1}w_{s_1,l_1}\overline{w}_{k_2,r_2}w_{s_2,l_2}\{\overline{w}_{k_1,r_1},\varphi_{r_2,s_2}\} \\
	&\quad + \sum_{r_1=1}^{k_1}\sum_{s_1=l_1}^{m+1}\sum_{r_2=1}^{k_2}\sum_{s_2=l_2}^{m+1}\overline{w}_{k_1,r_1}\varphi_{r_1,s_1}\overline{w}_{k_2,r_2}w_{s_2,l_2}\{w_{s_1,l_1},\varphi_{r_2,s_2}\} \\
	&\quad + \sum_{r_1=1}^{k_1}\sum_{s_1=l_1}^{m+1}\sum_{r_2=1}^{k_2}\sum_{s_2=l_2}^{m+1}\overline{w}_{k_1,r_1}w_{s_1,l_1}\varphi_{r_2,s_2}w_{s_2,l_2}\{\varphi_{r_1,s_1},\overline{w}_{k_2,r_2}\} \\
	&\quad + \sum_{r_1=1}^{k_1}\sum_{s_1=l_1}^{m+1}\sum_{r_2=1}^{k_2}\sum_{s_2=l_2}^{m+1}\overline{w}_{k_1,r_1}w_{s_1,l_1}\overline{w}_{k_2,r_2}\varphi_{r_2,s_2}\{\varphi_{r_1,s_1},w_{s_2,l_2}\} \\
	&\quad + \sum_{r_1=1}^{k_1}\sum_{s_1=l_1}^{m+1}\sum_{r_2=1}^{k_2}\sum_{s_2=l_2}^{m+1}\overline{w}_{k_1,r_1}w_{s_1,l_1}\overline{w}_{k_2,r_2}w_{s_2,l_2}\{\varphi_{r_1,s_1},\varphi_{r_2,s_2}\}.
\end{split}\end{equation}
The second term of the right-hand side of \eqref{Eq:Poisson_Borel_Thm_Pf_1} is rewritten as
\begin{equation}\label{Eq:Poisson_Borel_Thm_Pf_2}\begin{split}
	&\sum_{r_1=1}^{k_1}\sum_{s_1=l_1}^{m+1}\sum_{r_2=1}^{k_2}\sum_{s_2=l_2}^{m+1}\overline{w}_{k_1,r_1}\varphi_{r_1,s_1}\overline{w}_{k_2,r_2}w_{s_2,l_2}\{w_{s_1,l_1},\varphi_{r_2,s_2}\} \\
	&= -\sum_{r_1=1}^{k_1}\sum_{s_1=\max\{l_1+1,l_2\}}^{m+1}\overline{w}_{k_1,r_1}\varphi_{r_1,s_1}w_{s_1,l_2}\sum_{r_2=1}^{\min\{s_2-1,k_2\}}\overline{w}_{k_2,r_2}\{\varphi_{r_2,s_1},w_{s_1,l_1}\} \\
	&= (n+1)\sum_{r=1}^{k_1}\sum_{s=\max\{l_1+1,l_2\}}^{m+1}\overline{w}_{k_1,r}\varphi_{r,s}w_{s,l_2}\sum_{t=l_1}^{\min\{s-1,k_2\}}\overline{w}_{k_2,t}w_{t,l_1} \\
	&= (n+1)\sum_{r=1}^{k_1}\sum_{s=k_2}^{m+1}\overline{w}_{k_1,r}\varphi_{r,s}w_{s,l_2}\sum_{t=l_1}^{k_2}\overline{w}_{k_2,t}w_{t,l_1} - (n+1)\sum_{r=1}^{k_1}\sum_{s=l_2}^{m+1}\overline{w}_{k_1,r}\varphi_{r,s}w_{s,l_2}\sum_{t=s}^{k_2}\overline{w}_{k_2,t}w_{t,l_1} \\
	&= (n+1)\delta_{k_2,l_1}\sum_{r=1}^{k_1}\sum_{s=l_2}^{m+1}\overline{w}_{k_1,r}\varphi_{r,s}w_{s,l_2} - (n+1)\sum_{r=1}^{k_1}\sum_{s=l_2}^{m+1}\overline{w}_{k_1,r}\varphi_{r,s}w_{s,l_2}\sum_{t=s}^{k_2}\overline{w}_{k_2,t}w_{t,l_1}.
\end{split}\end{equation}
Similarly, the third term of the right-hand side of \eqref{Eq:Poisson_Borel_Thm_Pf_1} is rewritten as
\begin{equation}\label{Eq:Poisson_Borel_Thm_Pf_3}\begin{split}
	&\sum_{r_1=1}^{k_1}\sum_{s_1=l_1}^{m+1}\sum_{r_2=1}^{k_2}\sum_{s_2=l_2}^{m+1}\overline{w}_{k_1,r_1}w_{s_1,l_1}\varphi_{r_2,s_2}w_{s_2,l_2}\{\varphi_{r_1,s_1},\overline{w}_{k_2,r_2}\} \\
	&= (n+1)\delta_{k_2,l_1}\sum_{r=1}^{k_1}\sum_{s=l_2}^{m+1}\overline{w}_{k_1,r}\varphi_{r,s}w_{s,l_2} - (n+1)\sum_{r=1}^{k_1}\sum_{s=l_2}^{m+1}\overline{w}_{k_1,r}\varphi_{r,s}w_{s,l_2}\sum_{t=l_1}^{r}\overline{w}_{k_2,t}w_{t,l_1}.
\end{split}\end{equation}
On the other hand, the fifth term of the right-hand side of \eqref{Eq:Poisson_Borel_Thm_Pf_1} is rewritten as
\begin{equation}\label{Eq:Poisson_Borel_Thm_Pf_4}\begin{split}
	&\sum_{r_1=1}^{k_1}\sum_{s_1=l_1}^{m+1}\sum_{r_2=1}^{k_2}\sum_{s_2=l_2}^{m+1}\overline{w}_{k_1,r_1}w_{s_1,l_1}\overline{w}_{k_2,r_2}w_{s_2,l_2}\{\varphi_{r_1,s_1},\varphi_{r_2,s_2}\} \\
	&= \sum_{r_1=1}^{k_1}\sum_{s_2=l_2}^{m+1}\overline{w}_{k_1,r_1}w_{s_2,l_2}\sum_{s_1=\max\{r_1+1,l_1\}}^{\min\{s_2-1,k_2\}}w_{s_1,l_1}\overline{w}_{k_2,s_1}\varphi_{r_1,s_2} \\
	&\quad - \sum_{r_2=1}^{k_2}\sum_{s_1=l_1}^{m+1}\overline{w}_{k_2,r_2}w_{s_1,l_1}\sum_{r_1=\max\{r_2+1,l_2\}}^{\min\{s_2-1,k_1\}}\overline{w}_{k_1,r_1}w_{r_1,l_2}\varphi_{r_2,s_2}.
\end{split}\end{equation}
The first term of the right-hand side of \eqref{Eq:Poisson_Borel_Thm_Pf_4} is rewritten as
\begin{equation}\label{Eq:Poisson_Borel_Thm_Pf_5}\begin{split}
	&\sum_{r_1=1}^{k_1}\sum_{s_2=l_2}^{m+1}\overline{w}_{k_1,r_1}w_{s_2,l_2}\sum_{s_1=\max\{r_1+1,l_1\}}^{\min\{s_2-1,k_2\}}w_{s_1,l_1}\overline{w}_{k_2,s_1}\varphi_{r_1,s_2} \\
	&= - (n+1)\sum_{r=1}^{k_1}\sum_{s=l_2}^{m+1}\overline{w}_{k_1,r}\varphi_{r,s}w_{s,l_2}\sum_{t=\max\{r+1,l_1\}}^{\min\{s-1,k_2\}}\overline{w}_{k_2,t}w_{t,l_1} \\
	&= -(n+1)\delta_{k_2,l_1}\sum_{r=1}^{k_1}\sum_{s=l_2}^{m+1}\overline{w}_{k_1,r}\varphi_{r,s}w_{s,l_2} \\
	&\quad + (n+1)\sum_{r=1}^{k_1}\sum_{s=l_2}^{m+1}\overline{w}_{k_1,r}\varphi_{r,s}w_{s,l_2}\sum_{t=l_1}^{r}\overline{w}_{k_2,t}w_{t,l_1} + (n+1)\sum_{r=1}^{k_1}\sum_{s=l_2}^{m+1}\overline{w}_{k_1,r}\varphi_{r,s}w_{s,l_2}\sum_{t=s}^{k_2}\overline{w}_{k_2,t}w_{t,l_1}.
\end{split}\end{equation}
Then we obtain
\begin{equation}\label{Eq:Poisson_Borel_Thm_Pf_6}
	(\text{RHS of \eqref{Eq:Poisson_Borel_Thm_Pf_2}}) + (\text{RHS of \eqref{Eq:Poisson_Borel_Thm_Pf_3}}) + (\text{RHS of \eqref{Eq:Poisson_Borel_Thm_Pf_5}}) = (n+1)\delta_{k_2,l_1}\sum_{r=1}^{k_1}\sum_{s=l_2}^{m+1}\overline{w}_{k_1,r}\varphi_{r,s}w_{s,l_2}.
\end{equation}

In a similar manner, we obtain
\begin{equation}\label{Eq:Poisson_Borel_Thm_Pf_7}\begin{split}
	&\sum_{r_1=1}^{k_1}\sum_{s_1=l_1}^{m+1}\sum_{r_2=1}^{k_2}\sum_{s_2=l_2}^{m+1}\overline{w}_{k_1,r_1}w_{s_1,l_1}\overline{w}_{k_2,r_2}\varphi_{r_2,s_2}\{\varphi_{r_1,s_1},w_{s_2,l_2}\} \\
	&+ \sum_{r_1=1}^{k_1}\sum_{s_1=l_1}^{m+1}\sum_{r_2=1}^{k_2}\sum_{s_2=l_2}^{m+1}\varphi_{r_1,s_1}w_{s_1,l_1}\overline{w}_{k_2,r_2}w_{s_2,l_2}\{\overline{w}_{k_1,r_1},\varphi_{r_2,s_2}\} \\
	&- \sum_{r_2=1}^{k_2}\sum_{s_1=l_1}^{m+1}\overline{w}_{k_2,r_2}w_{s_1,l_1}\sum_{r_1=\max\{r_2+1,l_2\}}^{\min\{s_2-1,k_1\}}\overline{w}_{k_1,r_1}w_{r_1,l_2}\varphi_{r_2,s_2} \\
	&= -(n+1)\delta_{k_1,l_2}\sum_{r=1}^{k_2}\sum_{s=l_1}^{m+1}\overline{w}_{k_2,r}\varphi_{r,s}w_{s,l_1}.
\end{split}\end{equation}
from the fourth and first terms of the right-hand side of \eqref{Eq:Poisson_Borel_Thm_Pf_1} and the second term of the right-hand side of \eqref{Eq:Poisson_Borel_Thm_Pf_4}.
Relations \eqref{Eq:Poisson_Borel_Thm_Pf_6} and \eqref{Eq:Poisson_Borel_Thm_Pf_7} imply
\begin{align*}
	\{(E^{j}_{l_1,k_1}|\widehat{U}),(E^{j}_{l_2,k_2}|\widehat{U})\} &= (n+1)\delta_{k_2,l_1}\sum_{r=1}^{k_1}\sum_{s=l_2}^{m+1}\overline{w}_{k_1,r}\varphi_{r,s}w_{s,l_2} - (n+1)\delta_{k_1,l_2}\sum_{r=1}^{k_2}\sum_{s=l_1}^{m+1}\overline{w}_{k_2,r}\varphi_{r,s}w_{s,l_1}. \\
	&= (n+1)\delta_{k_2,l_1}(E^{j}_{l_2,k_1}|W_0^{-1}\mathcal{M}W_0) - (n+1)\delta_{k_1,l_2}(E^{j}_{l_1,k_2}|W_0^{-1}\mathcal{M}W_0) \\
	&= (n+1)\delta_{k_2,l_1}(E^{j}_{l_2,k_1}|\widehat{U}) - (n+1)\delta_{k_1,l_2}(E^{j}_{l_1,k_2}|\widehat{U}) \\
	&= (n+1)([E^{j}_{l_2,k_2},E^{j}_{l_1,k_1}]|\widehat{U}).
\end{align*}
Hence, we have proved the theorem.
\end{proof}

\begin{cor}
System \eqref{Eq:DS_SR_Borel1} is expressed as a Hamiltonian system
\begin{equation}\label{Eq:DS_SR_Borel_Ham}\begin{split}
	\partial_i(\varphi^{j}_{k,l}) = \{H_i,\varphi^{j}_{k,l}\},\quad
	\partial_i(w^{j}_{l,k}) = \{H_i,w^{j}_{l,k}\},\quad
	H_i = \frac{(U_i|\widehat{U})}{2(n+1)} = \frac{(\mathcal{B}_i-\partial_i(W_0)W_0^{-1}|\mathcal{M})}{2(n+1)} \\
	(i=1,\ldots,m+1,\ j=0,\ldots,n,\ k,l=1,\ldots,m+1,\ k<l).
\end{split}\end{equation}
\end{cor}

\begin{rem}
The definition of the variables $\varphi^{j}_{k,l}$ and $w^{j}_{l,k}$ is motivated by \cite{Suz1}.
In \cite{NouY3}, Noumi and Yamada proposed a Poisson bracket for the $\mathfrak{b}_+$-valued function $\mathcal{M}$, given by
\[
	\{(x|\mathcal{M}),(y|\mathcal{M})\} = (n+1)([x,y]|\mathcal{M})\quad (x,y\in\mathfrak{n}_{-}).
\]
This Poisson bracket is equivalent to \eqref{Eq:Poisson_Borel}, as can be shown by a direct calculation.
While one might expect to derive \eqref{Eq:Poisson} from this Poisson bracket via a symplectic $2$-form, it is not yet clear how to handle the grading parameter $z$ within that framework.
Therefore, we provide a direct proof using the inverse matrix $W^{-1}$.
\end{rem}

\subsection{Hamiltonian system in terms of canonical coordinates}\label{Subsec:Cano_coor}

Let
\[
	\mu^{j}_{k,l} = -\varphi^{j}_{k,l} - \sum_{r=1}^{k-1}\overline{w}^{j}_{k,r}\varphi^{j}_{r,l},\quad
	\lambda^{j}_{l,k} = w^{j}_{l,k}\quad (j=0,\ldots,n,\ k,l=1,\ldots,m+1,\ k<l).
\]
Then they turn out to be canonical coordinates for Poisson bracket \eqref{Eq:Poisson_Borel}.

\begin{prop}
The relations
\begin{align*}
	\{\mu^{j_1}_{k_1,l_1},\lambda^{j_2}_{l_2,k_2}\} = (n+1)\delta_{j_1,j_2}\delta_{k_1,k_2}\delta_{l_1,l_2},\quad
	\{\mu^{j_1}_{k_1,l_1},\mu^{j_2}_{k_2,l_2}\} = 0,\quad
	\{\lambda^{j_1}_{l_1,k_1},\lambda^{j_2}_{l_2,k_2}\} = 0 \\
	(j_1,j_2=0,\ldots,n,\ k_1,k_2,l_1,l_2=1,\ldots,m+1,\ k_1<l_1,\ k_2<l_2)
\end{align*}
hold.
\end{prop}

\begin{proof}
As in the proof of the previous theorem, we omit the index $j$ of $\varphi^{j}_{k,l}$, $w^{j}_{l,k}$ and $\overline{w}^{j}_{l,k}$.
We also temporarily use the notation
\[
	\varphi_{k,l} = 0,\quad
	w_{l,k} = 0,\quad
	\overline{w}_{l,k} = 0\quad (l<k).
\]
The first relation is shown as
\begin{align*}
	\{\mu_{k_1,l_1},\lambda_{l_2,k_2}\} &= -\{\varphi_{k_1,l_1},w_{l_2,k_2}\} - \sum_{r=1}^{k_1-1}\{\overline{w}_{k_1,r}\varphi_{r,l_1},w_{l_2,k_2}\} \\
	&= (n+1)\delta_{l_1,l_2}\left(w_{k_1,k_2}+\sum_{r=k_2}^{k_1-1}\overline{w}_{k_1,r}w_{r,k_2}\right) \\
	&= (n+1)\delta_{k_1,k_2}\delta_{l_1,l_2}.
\end{align*}

We show the second relation.
It is rewritten as
\begin{equation}\label{Eq:Cano_coor_Lem_Pf_1}\begin{split}
	\{\mu_{k_1,l_1},\mu_{k_2,l_2}\} &= \{\varphi_{k_1,l_1},\varphi_{k_2,l_2}\} + \sum_{r=1}^{k_1-1}\{\overline{w}_{k_1,r}\varphi_{r,l_1},\varphi_{k_2,l_2}\} + \sum_{r=1}^{k_2-1}\{\varphi_{k_1,l_1},\overline{w}_{k_2,r}\varphi_{r,l_2}\} \\
	&\quad + \sum_{r_1=1}^{k_1-1}\sum_{r_2=1}^{k_2-1}\{\overline{w}_{k_1,r_1}\varphi_{r_1,l_1},\overline{w}_{k_2,r_2}\varphi_{r_2,l_2}\}.
\end{split}\end{equation}
The first term of the right-hand side of \eqref{Eq:Cano_coor_Lem_Pf_1} is rewritten as
\begin{align*}
	\{\varphi_{k_1,l_1},\varphi_{k_2,l_2}\} &= -(n+1)\delta_{l_1,k_2}\varphi_{k_1,l_2} + (n+1)\delta_{l_2,k_1}\varphi_{k_2,l_1}.
\end{align*}
The second term of the right-hand side of \eqref{Eq:Cano_coor_Lem_Pf_1} is rewritten as
\begin{equation}\label{Eq:Cano_coor_Lem_Pf_2}\begin{split}
	&\sum_{r=1}^{k_1-1}\{\overline{w}_{k_1,r}\varphi_{r,l_1},\varphi_{k_2,l_2}\} \\
	&= -(n+1)\sum_{r=1}^{k_1-1}\delta_{k_2,r}\overline{w}_{k_1,l_2}\varphi_{r,l_1} - (n+1)\sum_{r=1}^{k_1-1}\delta_{l_1,k_2}\overline{w}_{k_1,r}\varphi_{r,l_2} + (n+1)\sum_{r=1}^{k_1-1}\delta_{l_2,r}\overline{w}_{k_1,r}\varphi_{k_2,l_1} \\
	&= -(n+1)\overline{w}_{k_1,l_2}\sum_{r=1}^{l_1}\delta_{k_2,r}\varphi_{r,l_1} + (n+1)\overline{w}_{k_1,l_2}\sum_{r=k_1}^{l_1}\delta_{k_2,r}\varphi_{r,l_1} - (n+1)\delta_{l_1,k_2}\sum_{r=1}^{k_1-1}\overline{w}_{k_1,r}\varphi_{r,l_2} \\
	&\quad + (n+1)\varphi_{k_2,l_1}\sum_{r=1}^{k_1}\delta_{l_2,r}\overline{w}_{k_1,r} - (n+1)\delta_{l_2,k_1}\varphi_{k_2,l_1}.
\end{split}\end{equation}
If $l_1<k_2$ or $k_1<l_2$, then the fourth or first term of the right-hand side of \eqref{Eq:Cano_coor_Lem_Pf_2} vanishes since $\varphi_{k_2,l_1}=0$ or $\overline{w}_{k_1,l_2}=0$, respectively.
Otherwise, we have
\[
	-(n+1)\overline{w}_{k_1,l_2}\sum_{r=1}^{l_1}\delta_{k_2,r}\varphi_{r,l_1} + (n+1)\varphi_{k_2,l_1}\sum_{r=1}^{k_1}\delta_{l_2,r}\overline{w}_{k_1,r} = 0.
\]
If $k_1\leq k_2\leq l_1$, then the second term of the right-hand side of \eqref{Eq:Cano_coor_Lem_Pf_2} vanishes since $\overline{w}_{k_1,l_2}=0$.
Hence, we obtain
\[
	\sum_{r=1}^{k_1-1}\{\overline{w}_{k_1,r}\varphi_{r,l_1},\varphi_{k_2,l_2}\} = -(n+1)\delta_{l_1,k_2}\sum_{r=1}^{k_1-1}\overline{w}_{k_1,r}\varphi_{r,l_2} - (n+1)\delta_{l_2,k_1}\varphi_{k_2,l_1}.
\]
Similarly, the third term of the right-hand side of \eqref{Eq:Cano_coor_Lem_Pf_1} is rewritten as
\[
	\sum_{r=1}^{k_2-1}\{\varphi_{k_1,l_1},\overline{w}_{k_2,r}\varphi_{r,l_2}\} = (n+1)\delta_{l_2,k_1}\sum_{r=1}^{k_2-1}\overline{w}_{k_2,r}\varphi_{r,l_1} + (n+1)\delta_{l_1,k_2}\varphi_{k_1,l_2}.
\]
The fourth term of the right-hand side of \eqref{Eq:Cano_coor_Lem_Pf_1} is rewritten as
\begin{equation}\label{Eq:Cano_coor_Lem_Pf_3}\begin{split}
	&\sum_{r_1=1}^{k_1-1}\sum_{r_2=1}^{k_2-1}\{\overline{w}_{k_1,r_1}\varphi_{r_1,l_1},\overline{w}_{k_2,r_2}\varphi_{r_2,l_2}\} \\
	&= -(n+1)\sum_{r=1}^{\min\{k_1-1,k_2-1\}}\overline{w}_{k_2,r}\overline{w}_{k_1,l_2}\varphi_{r,l_1} + (n+1)\sum_{r=1}^{\min\{k_1-1,k_2-1\}}\overline{w}_{k_1,r}\overline{w}_{k_2,l_1}\varphi_{r,l_2} \\
	&\quad - (n+1)\sum_{r_1=1}^{k_1-1}\sum_{r_2=1}^{k_2-1}\delta_{l_1,r_2}\overline{w}_{k_1,r_1}\overline{w}_{k_2,r_2}\varphi_{r_1,l_2} + (n+1)\sum_{r_1=1}^{k_1-1}\sum_{r_2=1}^{k_2-1}\delta_{l_2,r_1}\overline{w}_{k_1,r_1}\overline{w}_{k_2,r_2}\varphi_{r_2,l_1} \\
	&= -(n+1)\overline{w}_{k_1,l_2}\sum_{r=1}^{\min\{k_1-1,k_2-1\}}\overline{w}_{k_2,r}\varphi_{r,l_1} + (n+1)\overline{w}_{k_2,l_1}\sum_{r=1}^{\min\{k_1-1,k_2-1\}}\overline{w}_{k_1,r}\varphi_{r,l_2} \\
	&\quad - (n+1)\overline{w}_{k_2,l_1}\sum_{r=1}^{k_1-1}\overline{w}_{k_1,r}\varphi_{r,l_2} + (n+1)\delta_{l_1,k_2}\sum_{r=1}^{k_1-1}\overline{w}_{k_1,r}\varphi_{r,l_2} \\
	&\quad + (n+1)\overline{w}_{k_1,l_2}\sum_{r=1}^{k_2-1}\overline{w}_{k_2,r}\varphi_{r,l_1} - (n+1)\delta_{l_2,k_1}\sum_{r=1}^{k_2-1}\overline{w}_{k_2,r}\varphi_{r,l_1}.
\end{split}\end{equation}
If $k_1<k_2$, then the first and fifth terms of the right-hand side of \eqref{Eq:Cano_coor_Lem_Pf_3} vanish since $\overline{w}_{k_1,l_2}=0$.
Otherwise, we have
\[
	-(n+1)\overline{w}_{k_1,l_2}\sum_{r=1}^{\min\{k_1-1,k_2-1\}}\overline{w}_{k_2,r}\varphi_{r,l_1} + (n+1)\overline{w}_{k_1,l_2}\sum_{r=1}^{k_2-1}\overline{w}_{k_2,r}\varphi_{r,l_1} = 0.
\]
The same holds for the second and third terms of the right-hand side of \eqref{Eq:Cano_coor_Lem_Pf_3}.
Hence, we obtain
\[
	\sum_{r_1=1}^{k_1-1}\sum_{r_2=1}^{k_2-1}\{\overline{w}_{k_1,r_1}\varphi_{r_1,l_1},\overline{w}_{k_2,r_2}\varphi_{r_2,l_2}\} = (n+1)\delta_{l_1,k_2}\sum_{r=1}^{k_1-1}\overline{w}_{k_1,r}\varphi_{r,l_2} - (n+1)\delta_{l_2,k_1}\sum_{r=1}^{k_2-1}\overline{w}_{k_2,r}\varphi_{r,l_1}.
\]
It follows that the right-hand side of \eqref{Eq:Cano_coor_Lem_Pf_1} vanishes.

The third relation follows from \eqref{Eq:Poisson_Borel}.
Hence, we have proved the proposition.
\end{proof}

These canonical coordinates satisfy an $m(m+1)(n+1)$-dimensional Hamiltonian system, which contains $m$ invariants.
In other words, Poisson bracket \eqref{Eq:Poisson_Borel} has $m$ Casimir functions.
Hence, we can reduce the dimensions by $2m$.

\begin{lem}\label{Lem:Cano_coor_inv}
The relations
\[
	\sum_{j=0}^{n}\sum_{l=1}^{i-1}\lambda^{j}_{i,l}\mu^{j}_{l,i} - \sum_{j=0}^{n}\sum_{k=i+1}^{m+1}\lambda^{j}_{k,i}\mu^{j}_{i,k} + (n+1)\eta_i = 0\quad (i=1,\ldots,m)
\]
hold, where $\eta_i$ are constants.
\end{lem}

\begin{proof}
We temporarily use the notation
\[
	\overline{K}_i = -\sum_{j=1}^{i-1}jK_j + \sum_{j=i}^{m}(m-j+1)K_j\quad (i=1,\ldots,m).
\]
Recall that $K_j$ was defined at the beginning of Section \ref{Subsec:SR_g0}.
Note that
\[
	\overline{K}_i = (m+1)\sum_{j=0}^{n}E^{j}_{i,i} - I\quad (i=1,\ldots,m),
\]
where $I$ is the identity matrix.
For each $i=1,\ldots,m$, we consider the normalized invariant form
\[
	(\overline{K}_i|\widehat{U}) = (W_0\overline{K}_iW_0^{-1}|\mathcal{M}).
\]
Then we have
\[
	W_0\overline{K}_iW_0^{-1} = \overline{K}_i + (m+1)\sum_{j=0}^{n}\left(\sum_{k=i+1}^{m+1}\sum_{l=1}^{i-1}w^{j}_{k,i}\overline{w}^{j}_{i,l}E^{j}_{k,l}+\sum_{l=1}^{i-1}\overline{w}^{j}_{i,l}E^{j}_{i,l}+\sum_{k=i+1}^{m+1}w^{j}_{k,i}E^{j}_{k,i}\right).
\]
This implies that
\begin{align*}
	&(W_0\overline{K}_iW_0^{-1}|\mathcal{M}) \\
	&= (\overline{K}_i|\mathcal{M}) + (m+1)\sum_{j=0}^{n}\left(\sum_{k=i+1}^{m+1}\sum_{l=1}^{i-1}w^{j}_{k,i}\overline{w}^{j}_{i,l}(E^{j}_{k,l}|\mathcal{M})+\sum_{l=1}^{i-1}\overline{w}^{j}_{i,l}(E^{j}_{i,l}|\mathcal{M})+\sum_{k=i+1}^{m+1}w^{j}_{k,i}(E^{j}_{k,i}|\mathcal{M})\right) \\
	&= (m+1)\sum_{j=0}^{n}(E^{j}_{i,i}|\mathcal{M}) + (m+1)\sum_{j=0}^{n}\left(\sum_{k=i+1}^{m+1}w^{j}_{k,i}\sum_{l=1}^{i-1}\overline{w}^{j}_{i,l}\varphi^{j}_{l,k}+\sum_{l=1}^{i-1}\overline{w}^{j}_{i,l}\varphi^{j}_{l,i}+\sum_{k=i+1}^{m+1}w^{j}_{k,i}\varphi^{j}_{i,k}\right) \\
	&= -(m+1)\sum_{j=0}^{n}(\kappa_{i-1+j(m+1)}-\kappa_{i+j(m+1)}) - (m+1)\sum_{j=0}^{n}\sum_{k=i+1}^{m+1}\lambda^{j}_{k,i}\mu^{j}_{i,k} + (m+1)\sum_{j=0}^{n}\sum_{l=1}^{i-1}\overline{w}^{j}_{i,l}\varphi^{j}_{l,i}.
\end{align*}
Note that $(I|\mathcal{M})=0$.
Furthermore, we have
\begin{align*}
	\sum_{l=1}^{i-1}\lambda^{j}_{i,l}\mu^{j}_{l,i} &= -w^{j}_{i,1}\varphi^{j}_{1,i} - \sum_{l=2}^{i-1}w^{j}_{i,l}\left(\varphi^{j}_{l,i}+\sum_{r=1}^{l-1}\overline{w}^{j}_{l,r}\varphi^{j}_{r,i}\right) \\
	&= -\sum_{l=1}^{i-1}w^{j}_{i,l}\varphi^{j}_{l,i} - \sum_{r=1}^{i-2}\varphi^{j}_{r,i}\sum_{l=r+1}^{i-1}w^{j}_{i,l}\overline{w}^{j}_{l,r} \\
	&= -\sum_{l=1}^{i-1}w^{j}_{i,l}\varphi^{j}_{l,i} + \sum_{r=1}^{i-2}w^{j}_{i,r}\varphi^{j}_{r,i} + \sum_{r=1}^{i-2}\overline{w}^{j}_{i,r}\varphi^{j}_{r,i} \\
	&= -w^{j}_{i,i-1}\varphi^{j}_{i-1,i} + \sum_{l=1}^{i-1}\overline{w}^{j}_{i,l}\varphi^{j}_{l,i} - \overline{w}^{j}_{i,i-1}\varphi^{j}_{i-1,i} \\
	&= \sum_{l=1}^{i-1}\overline{w}^{j}_{i,l}\varphi^{j}_{l,i}
\end{align*}
for $i>2$.
Note that $\overline{w}^{j}_{i,i-1}=-w^{j}_{i,i-1}$.
On the other hand, relation \eqref{Eq:DS3} implies that
\begin{align*}
	(\overline{K}_i|\widehat{U}) &= (\overline{K}_i|\rho) \\
	&= (m+1)\sum_{j=0}^{n}(E^{j}_{i,i}|\rho) \\
	&= -(m+1)(n+1)(\rho_{i-1}-\rho_i),
\end{align*}
where $\rho_0=0$.
It follows that
\[
	\sum_{j=0}^{n}\sum_{l=1}^{i-1}\lambda^{j}_{i,l}\mu^{j}_{l,i} - \sum_{j=0}^{n}\sum_{k=i+1}^{m+1}\lambda^{j}_{k,i}\mu^{j}_{i,k} = -(n+1)(\rho_{i-1}-\rho_i) + \sum_{j=0}^{n}(\kappa_{i-1+j(m+1)}-\kappa_{i+j(m+1)}).
\]
Its left-hand side turns out to be the constant $-(n+1)\eta_i$.
Hence, we have proved the lemma.
\end{proof}

\section{Isomonodromy system}\label{Sec:IMDS}

In this section we focus on the case $(m,n)=(2,1)$.
Then the Hamiltonian system obtained in Section \ref{Subsec:Cano_coor} is expressed in terms of 12 canonical coordinates
\begin{align*}
	&\lambda_0 = \lambda^{0}_{3,1},\quad
	\lambda_1 = \lambda^{0}_{2,1},\quad
	\lambda_2 = \lambda^{0}_{3,2},\quad
	\lambda_3 = \lambda^{1}_{3,1},\quad
	\lambda_4 = \lambda^{1}_{2,1},\quad
	\lambda_5 = \lambda^{1}_{3,2}, \\
	&\mu_0 = \mu^{0}_{1,3},\quad
	\mu_1 = \mu^{0}_{1,2},\quad
	\mu_2 = \mu^{0}_{2,3},\quad
	\mu_3 = \mu^{1}_{1,3},\quad
	\mu_4 = \mu^{1}_{1,2},\quad
	\mu_5 = \mu^{1}_{2,3}
\end{align*}
with 2 constraints
\begin{align*}
	&\lambda_0\mu_0 + \lambda_1\mu_1 + \lambda_3\mu_3 + \lambda_4\mu_4 = 2\eta_1, \\
	&\lambda_0\mu_0 + \lambda_2\mu_2 + \lambda_3\mu_3 + \lambda_5\mu_5 = 2\eta_1 + 2\eta_2.
\end{align*}
Hence, the Hamiltonian system is essentially 8-dimensional.
We reduce it to the isomonodromy system with spectral type $31,22,211,1111$, which is a 6-dimensional system of ordinary differential equations.

\begin{lem}\label{Lem:Cano_coor_red}
If we assume that $t_2=t_3=1$, then the relation
\[
	\mu_2 + \mu_5 = 0
\]
holds.
\end{lem}

\begin{proof}
We have
\begin{align*}
	\mu_2 &= -\varphi^{0}_{2,3} - \overline{w}^{0}_{2,1}\varphi^{0}_{1,3} \\
	&= -\varphi^{0}_{2,3} + w^{0}_{2,1}\varphi^{0}_{1,3} \\
	&= -(E_{3,2}|\mathcal{M}) + (E_{1,2}|W_0)(E_{3,1}|\mathcal{M}) \\
	&= -(E_{3,2}|W_0\widehat{U}W_0^{-1}) + (E_{1,2}|W_0)(E_{3,1}|W_0\widehat{U}W_0^{-1}) \\
	&= -(E_{3,2}|\widehat{U}) - (E_{1,2}|W_0)(E_{3,1}|\widehat{U}) + (E_{1,2}|W_0)(E_{3,1}|\widehat{U}) \\
	&= -(E_{3,2}|\widehat{U}).
\end{align*}
Recall that
\[
	U_i = \mathrm{ad}\,w_1(\Lambda_i)\quad (i=1,2,3).
\]
Then we obtain
\begin{align*}
	\mu_2 &= -(E_{3,2}|\rho+t_1U_1+U_2+U_3) \\
	&= -(E_{3,2}|t_1[w_1,\Lambda_1]+[w_1,\Lambda_2]+[w_1,\Lambda_3]) \\
	&= -(t_1[E_{1,3}+E_{4,1},E_{3,2}]+[E_{2,5}+E_{5,2},E_{3,2}]+[zE_{3,6}+zE_{6,3},E_{3,2}]|w_1) \\
	&= -(E_{3,5}-zE_{6,2}|w_1).
\end{align*}
Similarly, we obtain
\[
	\mu_5 = (zE_{6,2}-E_{3,5}|w_1).
\]
Hence, we have proved the lemma.
\end{proof}

By Lemma \ref{Lem:Cano_coor_inv} and Lemma \ref{Lem:Cano_coor_red}, we obtain
\begin{align*}
	&d\mu_0\wedge d\lambda_0 + d\mu_1\wedge d\lambda_1 + d\mu_2\wedge d\lambda_2 + d\mu_3\wedge d\lambda_3 + d\mu_4\wedge d\lambda_4 + d\mu_5\wedge d\lambda_5 \\
	&= d\mu_0\wedge d\lambda_0 + d\mu_1\wedge d\lambda_1 + d\mu_2\wedge d\lambda_2 \\
	&\quad - d\frac{\lambda_0\mu_0+(\lambda_2-\lambda_5)\mu_2}{\lambda_3}\wedge d\lambda_3 - d\frac{\lambda_1\mu_1-(\lambda_2-\lambda_5)\mu_2}{\lambda_4}\wedge d\lambda_4 - d\mu_2\wedge d\lambda_5 \\
	&= d\lambda_3\mu_0\wedge d\frac{\lambda_0}{\lambda_3} + d\lambda_4\mu_1\wedge d\frac{\lambda_1}{\lambda_4} + d\frac{\lambda_3\mu_2}{\lambda_4}\wedge d\frac{\lambda_4(\lambda_2-\lambda_5)}{\lambda_3}.
\end{align*}
Hence, we can take
\[
	p_1 = \frac{\lambda_3\mu_0}{2t_1},\quad
	q_1 = \frac{t_1\lambda_0}{\lambda_3},\quad
	p_2 = \frac{\lambda_4\mu_1}{2t_1},\quad
	q_2 = \frac{t_1\lambda_1}{\lambda_4},\quad
	p_3 = \frac{\lambda_3\mu_2}{2\lambda_4},\quad
	q_3 = \frac{\lambda_4(\lambda_2-\lambda_5)}{\lambda_3}
\]
as canonical coordinates for Poisson bracket \eqref{Eq:Poisson_Borel}.
Indeed, they satisfy
\[
	\{p_i,q_j\} = \delta_{i,j},\quad
	\{p_i,p_j\} = 0,\quad
	\{q_i,q_j\} = 0\quad (i,j=1,2,3).
\]
We also set
\begin{align*}
	&\alpha_0 = \frac12(1+\kappa_5-2\kappa_0+\kappa_1),\quad
	\alpha_1 = \frac12(\kappa_0-2\kappa_1+\kappa_2),\quad
	\alpha_2 = \frac12(\kappa_1-2\kappa_2+\kappa_3), \\
	&\alpha_3 = \frac12(1+\kappa_2-2\kappa_3+\kappa_4),\quad
	\alpha_4 = \frac12(\kappa_3-2\kappa_4+\kappa_5),\quad
	\alpha_5 = \frac12(\kappa_4-2\kappa_5+\kappa_0), \\
	&\eta_1 = -\rho_1 - \frac12(\kappa_0-\kappa_1+\kappa_3-\kappa_4),\quad
	\eta_2 = \rho_1 - \rho_2 - \frac12(\kappa_0-\kappa_2+\kappa_3-\kappa_5),\quad
	t = t_1^2.
\end{align*}
Then we obtain the following theorem by a direct calculation.

\begin{thm}
The canonical coordinates $q_i,p_i\ (i=1,2,3)$ satisfy the 6-dimensional Hamiltonian system
\begin{equation}\label{Eq:Even4}
	\frac{dq_i}{dt} = \{H,q_i\},\quad
	\frac{dp_i}{dt} = \{H,p_i\}
\end{equation}
with the Hamiltonian
\begin{align*}
	H = \left.\frac{d}{dz}\frac{\mathrm{tr}((M_{6,0}+zM_{6,1})^4)}{4t(t-1)}\right|_{z\to0},
\end{align*}
where
\begin{align*}
	&M_{6,0} = \begin{pmatrix}
		0 & -p_2 & -p_1 & 1 & 0 & 0 \\
		0 & \alpha_1 & -p_3-q_2p_1 & q_2-1 & 1 & 0 \\
		0 & 0 & \alpha_1+\alpha_2 & q_1-q_3-1 & q_3 & 1 \\
		0 & 0 & 0 & \alpha_1+\alpha_2+\alpha_3 & q_2p_2-q_3p_3+\eta_2 & q_1p_1+q_3p_3-\eta_1-\eta_2 \\
		0 & 0 & 0 & 0 & \alpha_1+\alpha_2+\alpha_3+\alpha_4 & q_1p_1+(q_3+1)p_3-\eta_1-\eta_2 \\
		0 & 0 & 0 & 0 & 0 & \alpha_1+\alpha_2+\alpha_3+\alpha_4+\alpha_5 \\
	\end{pmatrix}, \\
	&M_{6,1} = \begin{pmatrix}
		0 & 0 & 0 & 0 & 0 & 0 \\
		0 & 0 & 0 & 0 & 0 & 0 \\
		0 & 0 & 0 & 0 & 0 & 0 \\
		t & 0 & 0 & 0 & 0 & 0 \\
		-q_2+t & 1 & 0 & 0 & 0 & 0 \\
		-q_1+q_2q_3+t & -q_3 & 1 & 0 & 0 & 0 \\
	\end{pmatrix}.
\end{align*}
\end{thm}

System \eqref{Eq:Even4} is derived from the deformation problem of a system of linear differential equations
\begin{equation}\label{Eq:Lax_Even4_6th}
	z\frac{d}{dz}\psi_6 = (M_{6,0}+M_{6,1}z)\psi_6.
\end{equation}
Note that the matrices $M_{6,0}$ and $M_{6,1}$ are obtained from $\mathcal{M}$ via a gauge transformation and a replacement of the independent variable.
System \eqref{Eq:Lax_Even4_6th} can be reduced to a Fuchsian system with spectral type $31,22,211,1111$.
We outline this reduction below.

Via a gauge transformation
\[
	\psi_6 = z\widetilde{\psi}_6,
\]
a Laplace transformation
\[
	\frac{d}{dz}\widetilde{\psi}_6 \to \zeta\widehat{\psi}_6,\quad
	z\widetilde{\psi}_6 \to -\frac{d}{d\zeta}\widehat{\psi}_6,
\]
and a M\"{o}bius transformation $\zeta\to z^{-1}$, system \eqref{Eq:Lax_Even4_6th} is transformed to
\[
	z\frac{d}{dz}\widehat{\psi}_6 = (I-zM_{6,1})^{-1}(I+M_{6,0})\widehat{\psi}_6,
\]
where $I$ is the identity matrix.
It is reduced to a system with $5\times5$ matrices
\begin{equation}\label{Eq:Lax_Even4_5th}
	z\frac{d}{dz}\psi_5 = (M_{5,0}+zM_{5,1})\psi_5,
\end{equation}
since the first column of the matrix $(I-zM_{6,1})^{-1}(I+M_{6,0})$ turns out to be the zero vector.
Furthermore, via a gauge transformation
\[
	\psi_5(z) \to z^{1+\alpha_1}\widetilde{\psi}_5(z),
\]
a Laplace transformation
\[
	\frac{d}{dz}\widetilde{\psi}_5 \to \zeta\widehat{\psi}_5,\quad
	z\widetilde{\psi}_5 \to -\frac{d}{d\zeta}\widehat{\psi}_5,
\]
and a replacement $\zeta\to z$, system \eqref{Eq:Lax_Even4_5th} is transformed to
\[
	\frac{d}{dz}\widehat{\psi}_5 = (M_{5,1}-zI)^{-1}(I+M_{5,0})\widehat{\psi}_5.
\]
It is reduced to a system with $4\times4$ matrices
\begin{equation}\label{Eq:Lax_Even4_4th}
	\frac{d}{dz}\psi_4 = \left(\frac{M_{4,t}}{z-t}+\frac{M_{4,1}}{z-1}+\frac{M_{4,0}}{z}\right)\psi_4,
\end{equation}
since the first column of the matrix $(M_{5,1}-zI)^{-1}(I+M_{5,0})$ turns out to be the zero vector.
System \eqref{Eq:Lax_Even4_4th} is Fuchsian and its Riemann scheme is given by
\[
	\left\{\begin{array}{cccc}
		z=t & z=1 & z=0 & z=\infty \\
		0 & 0 & 0 & \alpha_2 \\
		0 & 0 & 0 & \alpha_2+\alpha_3 \\
		0 & -\alpha_1-2\alpha_2-\alpha_3-\alpha_4-\alpha_5-\eta_1-\eta_2 & \alpha_1 & \alpha_2+\alpha_3+\alpha_4 \\
		-\alpha_2-\alpha_3+\eta_1 & -\alpha_1-\alpha_2-\alpha_3-\alpha_4+\eta_2 & \alpha_1 & \alpha_2+\alpha_3+\alpha_4+\alpha_5
	\end{array}\right\}.
\]
We can verify this fact by a direct calculation.

Finally, we propose a conjecture concerning the connection between the similarity reduction and isomonodromy systems.

\begin{conj}\label{Conj:DS_IMDS}
Hamiltonian system \eqref{Eq:DS_SR_Borel_Ham} reduces to the isomonodromy system with spectral type $\{(m(n+1)-1,1)\times(m+1),(m^{n+1}),(1^{m(n+1)})\}$.
\end{conj}

\section{Conclusion}

In this article, we have established a method to express the similarity reduction as a Hamiltonian system in terms of canonical coordinates for the conjugacy class $(n+1,\ldots,n+1)$.
We have also clarified the connection between the similarity reduction for the case $(m,n)=(2,1)$ and the isomonodromy system of type $31,22,211,1111$.
These results will serve as a foundation for future research on other conjugacy classes and affine Lie algebras.

Recently, the analysis of the Painlev\'e equations and isomonodromy systems has been extensively developed; see \cite{GIL} and references therein.
In these studies, tau functions play an important role and are often given as integrals of the Hamiltonians.
Therefore, it is essential to express those systems in Hamiltonian form.
Furthermore, since asymptotic solutions have been derived with the aid of conformal field theory, the quantization of these systems is an important subject; see \cite{AHKOSSY,N} and references therein for previous research.
Hence, finding suitable canonical coordinates is also essential.
The results of this article are expected to be useful in the analysis of isomonodromy systems, especially given their close connection to conformal field theory.



\end{document}